\documentclass[11pt]{article}

\usepackage[margin=1in]{geometry}
\usepackage{amsmath,amssymb,amsfonts,amsthm}
\usepackage{bm}
\usepackage{graphicx}
\usepackage{booktabs}
\usepackage{hyperref}
\usepackage{enumitem}
\usepackage{float}
\usepackage{mathtools}
\usepackage{url}

\hypersetup{
  colorlinks=true,
  linkcolor=blue,
  citecolor=blue,
  urlcolor=blue
}

\newtheorem{proposition}{Proposition}
\newtheorem{corollary}{Corollary}
\newtheorem{remark}{Remark}
\newtheorem{definition}{Definition}

\newcommand{\sfee}{\sigma_{\mathrm{fee}}}
\newcommand{\gbar}{\bar\gamma}

\title{Fee Implied Volatility for Concentrated Liquidity:\\ Theta, Loss Versus Rebalancing, and Identification}
\author{Amy Oumayma Khaldoun \\ \texttt{amy.khaldoun@protonmail.com}, \texttt{amy@sqv3.com}}
\date{July 2026, revised September 2026}

\begin{document}

\maketitle

\begin{abstract}
Narrow Uniswap v3 liquidity ranges behave like short dated options, and \href{https://panoptic.xyz/}{Panoptic}\footnote{Panoptic is a decentralized, oracle free perpetual options protocol built on top of Uniswap v3.} streams option premia that track the fees such ranges earn. This suggests reading LP fee income as option theta and inverting it for a volatility, as proposed by Lambert. We give this idea a precise derivation and a precise interpretation. Writing a concentrated liquidity position as a short strip of European options and using the short maturity concentration of Black Scholes theta at the strike, we show that the aggregate theta of active liquidity equals $\sigma^2 L\sqrt{P}/4$ per unit time, which is exactly its loss versus rebalancing (LVR). Matching this theta to realized fee income defines a fee implied volatility $\sfee^2 = 2\gamma V/(\Delta\bar\delta)$, where $\gamma$ is the fee rate, $V$ the volume over a window of length $\Delta$, and $\delta$ the numeraire depth per unit of log price. The definition is dimensionally consistent, annualized by construction, and reduces to the familiar tick liquidity formula when depth is measured over a band of log width $2\gamma$. The resulting $\sfee$ is the breakeven volatility of liquidity provision: the volatility at which fee income exactly offsets LVR. In a benchmark where arbitrageurs keep the pool price inside the no trade band set by the fee, $\sfee$ recovers the volatility of the external price up to a factor $1+O(\gamma)$, and the fee tier cancels. Uninformed flow biases $\sfee^2$ upward by the ratio of total to arbitrage volume, while frictions that widen the effective band bias it downward. A depth matched version of the estimator is a fee multiple of the realized total variation of the pool log price and requires no liquidity data. On the ETH/USDC 30 bps pool from 2023 to 2026, $\sfee$ has a correlation of 0.87 with realized volatility, and the implied ratio of fee income to LVR has a median of 1.18.
\end{abstract}


\newpage
\tableofcontents

\section{Introduction}\label{sec:intro}

Implied volatility is a key state variable in options markets. On centralized venues it is backed out from traded option premia by inverting a pricing model such as Black Scholes \cite{BlackScholes1973,Merton1973,Gatheral2006VolSurface}. That procedure presupposes quoted option prices, a notion of maturity, and a well defined mapping from the observed premium to a volatility parameter.

On Uniswap v3 and in Panoptic, the economic primitives are different. Liquidity is concentrated over price ranges rather than posted in a limit order book, options are implemented through relocations of liquidity, and Panoptic premia are streamed through time rather than paid entirely upfront \cite{Adams2021UniswapV3,Lambert2022Panoptic,PanopticDocsPerps}, part of a broader and fast growing landscape of on chain options infrastructure \cite{Castle2023DeFiOptions,Polygon2022OnChainOptions}. This raises a natural question: can one construct a DEX native analogue of implied volatility directly from on chain pool data, without any external oracle?

A concrete answer has circulated since 2021. Lambert observed that narrow Uniswap v3 ranges behave like perpetual options \cite{Lambert2021PerpetualOptions} and that equating their fee income to option theta yields $\sigma = 2\gamma\sqrt{V/D}$, with $\gamma$ the fee rate, $V$ daily volume and $D$ the liquidity in the current tick, annualized by $\sqrt{365}$ \cite{Lambert2021OnChainVol}. Panoptic uses this expression to report an implied volatility for its streamed premia, and an earlier version of the present construction appeared in \cite{KhaldounPanopticIV}. The formula is simple and widely used, but three things have been left implicit: the units in which $D$ must be measured for the expression to be well posed, the economic object the resulting $\sigma$ measures, and the conditions under which it coincides with the volatility of the external price.

This paper makes all three explicit. Our contributions are as follows.
\begin{enumerate}[leftmargin=2em]
  \item We write a concentrated liquidity position as a short strip of European options and use the short maturity concentration of Black Scholes theta to show that the aggregate theta of active liquidity is $\sigma^2L\sqrt{P}/4$ per unit time. This is exactly the loss versus rebalancing (LVR) of the same liquidity \cite{Milionis2022LVR}: theta and LVR are one function of volatility, evaluated at the pricing volatility and at the realized volatility respectively (Proposition~\ref{prop:lp_theta}).
  \item Matching this theta to realized fee income defines a fee implied volatility $\sfee^2 = 2\gamma V/(\Delta\bar\delta)$, with $\delta$ the numeraire depth per unit of log price. The definition is dimensionally consistent and annualized by construction, recovers the tick liquidity formula under an explicit normalization, and identifies $\sfee$ as the breakeven volatility of liquidity provision (Definition~\ref{def:fee_iv} and Proposition~\ref{prop:breakeven}).
  \item We characterize when $\sfee$ estimates the volatility $\sigma$ of the external price. In a pure arbitrage benchmark it does, up to a factor $1+O(\gamma)$, and the fee tier cancels (Proposition~\ref{prop:benchmark}). Uninformed flow biases it upward by the ratio of total to arbitrage volume, and arbitrage frictions bias it downward (Section~\ref{subsec:noise}). A depth matched version equals $2\gbar$ times the realized total variation of the pool log price per year and requires no liquidity data (Section~\ref{subsec:tv}).
  \item We illustrate on the ETH/USDC 30 bps pool from 2023 to 2026, where the implied ratio of fee income to LVR has a median of 1.18 (Section~\ref{sec:empirics}).
\end{enumerate}

For Panoptic, the result gives $\sfee$ a precise meaning: it is the volatility at which premia matched to fee income exactly compensate LVR, and its departure from external volatility measures how much LP income is paid by uninformed flow.

Section~\ref{sec:background} sets out the context and Section~\ref{sec:uniswapv3_context} the relevant Uniswap v3 mechanics. Section~\ref{sec:theta_dirac} develops the short maturity concentration of theta, Section~\ref{sec:strip} applies it to active liquidity, and Section~\ref{sec:fee_iv} defines the fee implied volatility. Section~\ref{sec:identification} studies identification, Section~\ref{sec:empirics} presents the empirical illustration, Section~\ref{sec:implementation} covers implementation, and Section~\ref{sec:scope} discusses scope and limits.

\section{Context}
\label{sec:background}

\subsection{Implied volatility in centralized markets}

In classical option pricing, the Black Scholes formula gives the value of a European call or put as a function of spot price $S$, strike $K$, maturity $T$, risk free rate $r$, and volatility $\sigma$. If one observes a market price $C_{\mathrm{mkt}}$ or $P_{\mathrm{mkt}}$, the implied volatility is defined by inversion:
\[
C_{\mathrm{BS}}(S,K,T,r,\sigma_{\mathrm{IV}})
=
C_{\mathrm{mkt}}.
\]
This definition is meaningful because the option premium is directly observed.

On AMM based DEXs, no such premium is quoted for a standard vanilla contract. The economic object that is observed is instead the trading and fee flow through a liquidity pool. Any DEX native volatility notion must therefore be defined relative to these observables rather than as a literal replacement for order book implied volatility.

The idea of estimating volatility from order flow rather than from quoted premia is not new to DEXs. A separate strand of market microstructure literature has studied volatility estimation directly from limit order book dynamics \cite{Bibinger2014LOBVol} and the relationship between trading volume and price volatility in traditional order driven markets \cite{Naes2006VolumeVol}. The construction in this paper can be read as an AMM native instance of that broader tradition, adapted to a setting where the relevant observable is pool fee flow rather than book depth or order arrival.

\subsection{Concentrated liquidity and Panoptic}

Uniswap v3 allows liquidity providers to choose a finite active interval $[P_a,P_b]$ and earn fees only while price remains in range \cite{Adams2021UniswapV3,UniswapDocsConcentrated}. Very narrow ranges behave like range orders and are economically close to option positions \cite{Hashemseresht2022Concentrated,Fan2023StrategicLP,Lambert2021PerpetualOptions}. Panoptic builds on this observation by treating Uniswap v3 LP positions as primitives for perpetual options \cite{Lambert2022Panoptic,PanopticDocsPerps}. Premiums are streamed through time as functions of price proximity and time in range rather than settled as a single upfront payment.

This architecture motivates looking for an observable quantity that links fee flow, liquidity, and a volatility like state variable.

\begin{remark}[A structural inversion relative to centralized options]
\label{rem:inversion}
In a centralized options market, the implied volatility at which a position is opened is locked into its premium at inception, while its gamma evolves in calendar time and grows sharply near expiry. Panoptic reverses the roles. A Panoptic position has no expiry, so its greeks are time homogeneous: at a given price its gamma does not change with calendar time, although it still depends on price and vanishes out of range. The premium, by contrast, is not locked at inception. It streams at a rate driven by current fee flow and a spread multiplier, so the implied quantity floats while the gamma profile stays fixed. This reversal is part of why CEX implied volatility machinery does not transplant directly, and why a flow based definition such as the one in Section~\ref{sec:fee_iv} is natural.
\end{remark}

\subsection{Fees and loss versus rebalancing as the two legs of LP accounting}
\label{subsec:two_legs}

A narrow LP position is short convexity. In Black Scholes accounting, a delta hedged short option earns theta at the volatility used to price it and pays gamma costs at the volatility the underlying realizes, so that over a short interval its profit is proportional to the difference between implied and realized variance.

The AMM has the same two legs. Fees play the role of theta: they are the compensation the LP receives for supplying convexity. Loss versus rebalancing plays the role of the gamma cost: it is the cost of carrying negative convexity, realized through arbitrage trades that correct stale pool prices \cite{Milionis2022LVR,Cartea2023PredictableLoss}. Fees and LVR are therefore not two parts of one premium; they are the premium and the cost, and the relevant question is at which volatility they balance. Sections~\ref{sec:strip} and~\ref{sec:fee_iv} show that the LP's theta, computed from its Black Scholes replication, coincides with its LVR as a function of volatility, so that matching fees to theta defines exactly this balancing volatility. Section~\ref{sec:identification} asks when it equals the volatility of the external price.

\section{Uniswap v3 context for a DEX native volatility}
\label{sec:uniswapv3_context}

This section summarizes the parts of the Uniswap v3 design that matter for the fee implied volatility. We focus on concentrated liquidity, ticks and square root prices, the relation between fee tiers and tick spacing, and the on chain observables used in the estimators.

\subsection{Constant product market making with concentrated liquidity}

Uniswap v3 belongs to the class of constant function market makers. Earlier versions of Uniswap implement a constant product market maker with reserves $(x,y)$ and invariant
\begin{equation}
x y = k.
\end{equation}
This provides liquidity over the entire price range $(0,\infty)$ but is capital inefficient because only a fraction of pooled assets are used near the prevailing price.

Uniswap v3 introduces \emph{concentrated liquidity}. Liquidity providers choose a finite price interval $[p_a,p_b]$ and supply liquidity only inside this range \cite{Adams2021UniswapV3,UniswapDocsConcentrated}. Each such range defines a position that behaves as a constant product pool with larger virtual reserves while the spot price stays inside the interval. Outside the range the position holds a single asset and no longer earns fees. The liquidity of a position is measured by a scalar $L=\sqrt{k}$ and the real reserves $(x,y)$ obey
\begin{equation}
\label{eq:univ3_real_reserves}
\left(x+\frac{L}{\sqrt{p_b}}\right)
\left(y+L\sqrt{p_a}\right)
=
L^2,
\end{equation}
which is a translation of the constant product curve that ensures solvency exactly inside $[p_a,p_b]$.

Positions on very narrow ranges act similarly to range orders. When the price crosses the range from one side to the other the reserves flip from entirely in one asset to entirely in the other, plus fees, which motivates the option like interpretation of narrow LP positions.

\subsection{Ticks, square root price, and liquidity}

To implement arbitrary ranges, Uniswap v3 discretizes the price space through integer tick indices $i$. Each tick corresponds to a price
\begin{equation}
p(i)=1.0001^i,
\end{equation}
so that consecutive ticks are separated by roughly one basis point.

In practice the protocol works in square root price space. The square root price at tick $i$ is
\begin{equation}
\sqrt{p(i)}=1.0001^{i/2},
\end{equation}
and the pool state stores a fixed point representation $\sqrt{P}$ of the current square root price and a scalar liquidity $L$. Between two adjacent initialized ticks the pool behaves as a constant product AMM with virtual reserves $(x,y)$ related to $(L,\sqrt{P})$ by
\begin{equation}
\label{eq:virtual_reserves}
x=\frac{L}{\sqrt{P}},
\qquad
y=L\sqrt{P}.
\end{equation}
Equivalently, liquidity can be understood as the sensitivity of the token one reserve to the square root price
\begin{equation}
L=\frac{\Delta y}{\Delta \sqrt{P}},
\end{equation}
which leads to simple swap formulas for small moves inside a tick \cite{Adams2021UniswapV3}.

The pool also tracks the current tick index $i_c$, defined as the greatest tick not exceeding the present square root price
\begin{equation}
i_c
=
\left\lfloor
\log_{\sqrt{1.0001}}\sqrt{P}
\right\rfloor.
\end{equation}
Liquidity changes only when the price crosses an initialized tick or when positions are minted or burned, never during swaps that remain inside an interval between initialized ticks.

\subsection{Tick spacing, ranges, and fee tiers}

Not every integer tick can be used as a range boundary. Each pool is created with a parameter \texttt{tickSpacing}, denoted $t_S$, and only ticks with indices that are multiples of $t_S$ can be initialized as lower or upper bounds of positions \cite{Adams2021UniswapV3}.

The initial Uniswap v3 configuration supports three fee rates, which we denote by $\gamma$ from here on,
\begin{equation}
\gamma\in\{0.0005,\ 0.0030,\ 0.01\},
\end{equation}
and associates them with tick spacings
\begin{equation}
t_S\in\{10,\ 60,\ 200\},
\end{equation}
respectively \cite{Adams2021UniswapV3}. These imply minimal widths of approximately $0.10\%$, $0.60\%$, and $2.02\%$ between initializable ticks. For these three tiers the relation
\begin{equation}
\label{eq:tS_fee_rel}
t_S = 2\cdot 10^4\,\gamma
\end{equation}
holds exactly. It is a design choice rather than a law: the 1 bp tier added later uses $t_S=1$, and Uniswap v4 lets each pool choose its fee and tick spacing independently, with hooks allowing the fee to vary over time. Nothing in the derivation below depends on \eqref{eq:tS_fee_rel}. Its only role is to explain why, for the original tiers, depth measured over one tick spacing is the right normalization (Section~\ref{subsec:tick_form}).

\subsection{Price and liquidity oracles}

Uniswap v3 improves on earlier designs by tracking accumulators of both price and liquidity within the core pool contracts \cite{Adams2021UniswapV3,UniswapDocsOracles}. Each pool maintains a seconds weighted accumulator of the logarithmic price, expressed as the tick index
\begin{equation}
a_t=\sum_{i=1}^{t}\log_{1.0001}(P_i),
\end{equation}
where $P_i$ is the instantaneous price at second $i$. The time weighted geometric mean price over a period $[t_1,t_2]$ can then be recovered as
\begin{equation}
P_{t_1,t_2}
=
1.0001^{\frac{a_{t_2}-a_{t_1}}{t_2-t_1}}.
\end{equation}

In parallel, the pool tracks a seconds weighted accumulator of the inverse liquidity
\begin{equation}
\mathrm{secondsPerLiquidityCumulative}(t)
=
\int_0^t \frac{1}{L_u}\,du,
\end{equation}
updated at the beginning of each block \cite{Adams2021UniswapV3}. These oracles are useful for reconstructing time in range and liquidity weighted quantities, but they do not by themselves reveal the decomposition between arbitrage and uninformed order flow.

\subsection{Observable quantities}
\label{subsec:observables}

The pool exposes or allows reconstruction of:
\begin{enumerate}[leftmargin=2em]
  \item the fee rate $\gamma$,
  \item active liquidity $L_t$ and the square root price $\sqrt{P_t}$, block by block, and hence the numeraire depth $\delta_t$ introduced in Section~\ref{sec:strip},
  \item swap events, each with its input amount and the pool price after execution.
\end{enumerate}

These suffice to compute the fee implied volatility of Section~\ref{sec:fee_iv} and its depth matched version in Section~\ref{subsec:tv}. Whether that quantity estimates the volatility of an external reference price is a separate question, addressed in Section~\ref{sec:identification}.

\section{Theta and Dirac type approximation}
\label{sec:theta_dirac}

\subsection{Notation}

We use the following notation throughout this section.
\begin{itemize}[leftmargin=2em]
  \item $S$, spot price of the underlying asset,
  \item $K$, strike price,
  \item $\sigma$, volatility parameter,
  \item $t$, time to maturity in the Black Scholes setting.
\end{itemize}

We work under zero interest rate and no dividends. Limits as $t\to 0$ are taken with $S$ close to $K$, consistent with the short dated option interpretation of narrow range liquidity.

\subsection{Theta of a short dated option}

For a European option under the Black Scholes model with zero rates and no dividends, theta can be written as
\[
\theta(S,t)
=
\frac{S\sigma}{\sqrt{8\pi t}}
\exp\left(
-
\frac{\left(\ln(S/K)+\frac{\sigma^2 t}{2}\right)^2}{2\sigma^2 t}
\right),
\]
up to the sign convention distinguishing long and short theta \cite{BlackScholes1973,Merton1973}. For our purposes only the concentration and total mass matter.

As $t$ decreases while $S$ stays near $K$, the Gaussian term becomes increasingly narrow in $\ln(S/K)$. The prefactor grows like $t^{-1/2}$, and the net effect is that $\theta(S,t)$ develops a sharp spike around $S=K$ whose integral over $S$ has a finite limit.

\subsection{Dirac distribution and Gaussian approximation}

The Dirac delta distribution is defined formally through its action on test functions. A classical approximation is given by centered Gaussians \cite{Lighthill1960}:
\[
\delta(x)
=
\lim_{\varepsilon\to 0}
\frac{1}{\varepsilon\sqrt{2\pi}}
\exp\left(
-\frac{x^2}{2\varepsilon^2}
\right)
\]
in the sense of distributions. For any smooth compactly supported test function $f$,
\[
\int_{\mathbb{R}}
f(x)\,
\frac{1}{\varepsilon\sqrt{2\pi}}
\exp\left(
-\frac{x^2}{2\varepsilon^2}
\right)\,dx
\longrightarrow
f(0)
\qquad
\text{as }\varepsilon\to 0.
\]

\subsection{Theta as a scaled Dirac impulse}

Set $x=\ln(S/K)$ and $\varepsilon^2=\sigma^2 t$. For small $t$, the exponent in theta satisfies
\[
\frac{\left(\ln(S/K)+\frac{\sigma^2 t}{2}\right)^2}{2\sigma^2 t}
\approx
\frac{x^2}{2\varepsilon^2}.
\]
Using $S=Ke^x$ and $dS=Ke^x\,dx$, one obtains the following distributional limit.

\begin{proposition}[Dirac approximation of theta]
\label{prop:dirac_theta}
As $t\to 0$, theta converges in the sense of distributions to a scaled Dirac mass located at $S=K$:
\[
\theta(S,t)\rightharpoonup \frac{K^2\sigma^2}{2}\,\delta(S-K).
\]
Equivalently, for every smooth compactly supported test function $f$,
\[
\int_0^\infty f(S)\,\theta(S,t)\,dS
\longrightarrow
\frac{K^2\sigma^2}{2}\,f(K).
\]
\end{proposition}

\begin{proof}
The proof is given in Appendix~\ref{app:theta}. It uses the change of variables $u=\ln(S/K)$, Gaussian localization, and dominated convergence.
\end{proof}

\begin{remark}[A one line route to the mass]\label{rem:mass}
At zero rates the Black Scholes equation gives $\lvert\theta\rvert = \tfrac12\sigma^2S^2\Gamma$, and a Gaussian integral gives $\int_0^\infty S^2\Gamma\,dS = K^2e^{\sigma^2t}$. The mass $\tfrac12K^2\sigma^2e^{\sigma^2t}$ computed in Appendix~\ref{app:theta} follows at once. What Proposition~\ref{prop:dirac_theta} adds is localization: the mass concentrates at the strike. Section~\ref{sec:strip} uses the same localization with the roles of spot and strike exchanged.
\end{remark}

\begin{proposition}[Concentration in the strike variable]\label{prop:dirac_strike}
Fix $S>0$ and let $\theta(S;K,t)$ denote the theta of a European option with spot $S$, strike $K$ and time to maturity $t$, at zero rates and without dividends. Then
\begin{equation}\label{eq:strike_mass}
\int_0^\infty \theta(S;K,t)\,dK = \tfrac12\sigma^2S^2e^{\sigma^2t}
\;\longrightarrow\;\tfrac12\sigma^2S^2
\qquad\text{as } t\to 0,
\end{equation}
and for every smooth compactly supported function $n$ on $(0,\infty)$,
\begin{equation}\label{eq:strike_dirac}
\int_0^\infty n(K)\,\theta(S;K,t)\,dK \;\longrightarrow\; \tfrac12\sigma^2S^2\,n(S)
\qquad\text{as } t\to 0.
\end{equation}
The same limit holds for piecewise smooth $n$ that is continuous at $S$.
\end{proposition}

\begin{proof}
Set $u=\ln(K/S)$, so that $K=Se^u$ and $dK=Se^u\,du$, and repeat the argument of Appendix~\ref{app:theta}. The Gaussian factor now localizes the integral at $K=S$, and the total mass follows from the same completion of the square. For piecewise smooth $n$ continuous at $S$, localization confines the limit to an arbitrarily small neighbourhood of $S$ on which $n$ is continuous.
\end{proof}

\section{Active liquidity as a strip of short options}
\label{sec:strip}

\subsection{Greeks of a concentrated position}

Let the pool trade a risky token against a numeraire, and let $P$ denote the price of the risky token in units of the numeraire, after adjusting for token decimals and, where necessary, for token ordering (Section~\ref{subsec:units}).

Consider a position with liquidity $L$ on $[P_a,P_b]$ with $P_a<P<P_b$. By \eqref{eq:univ3_real_reserves} and \eqref{eq:virtual_reserves}, its real reserves satisfy
\begin{equation}\label{eq:real_reserves}
x(P)+\frac{L}{\sqrt{P_b}}=\frac{L}{\sqrt{P}},
\qquad
y(P)+L\sqrt{P_a}=L\sqrt{P}.
\end{equation}
Let $V(P)=P\,x(P)+y(P)$ be its value in numeraire. Differentiating \eqref{eq:real_reserves} gives $V'(P)=x(P)$, the risky holding, and
\begin{equation}\label{eq:lp_gamma}
\lvert V''(P)\rvert=\frac{L}{2P^{3/2}},\qquad P_a<P<P_b,
\end{equation}
with $V''$ negative inside the range and zero outside it. The position is short gamma while in range and carries no gamma out of range.

\subsection{Static replication}

By the Carr Madan spanning formula \cite{CarrMadan2001}, a payoff with continuous first derivative and piecewise continuous second derivative is replicated by a static holding of the underlying and the numeraire together with a strip of European options whose strike density equals the second derivative of the payoff; see \cite{Clark2020Replicating} for the constant product case. Below the range the position holds only the risky token and above it only the numeraire, so $V'$ is continuous at $P_a$ and $P_b$ and the strip carries no point masses. The concentrated position is therefore equivalent, up to a static holding, to a short strip with density
\begin{equation}\label{eq:strip_density}
n(K)=\frac{L}{2K^{3/2}}\,\mathbf{1}_{[P_a,P_b]}(K).
\end{equation}
Following the perpetual option interpretation of Panoptic \cite{Lambert2022Panoptic}, we treat the position as continuously rolling this strip at vanishing maturity.

\subsection{Theta of active liquidity}

Applying Proposition~\ref{prop:dirac_strike} to \eqref{eq:strip_density} with $S=P$ in the interior of the range gives the following.

\begin{proposition}[Theta of active liquidity]\label{prop:lp_theta}
As the maturity of the strip tends to zero, its aggregate theta per unit time converges to
\begin{equation}\label{eq:lp_theta}
\Theta_{\mathrm{LP}}(\sigma)
=\tfrac12\sigma^2P^2\,\lvert V''(P)\rvert
=\frac{\sigma^2L\sqrt{P}}{4}.
\end{equation}
\end{proposition}

Expression \eqref{eq:lp_theta} is exactly the instantaneous loss versus rebalancing of the active liquidity, $\tfrac12\sigma^2P^2\lvert x'(P)\rvert$, derived by Milionis, Moallemi, Roughgarden and Zhang \cite{Milionis2022LVR}. The theta an LP collects in the Black Scholes accounting and the LVR it pays in the AMM accounting are the same function of volatility. They differ only in which volatility is plugged in: theta is evaluated at the volatility used for pricing, and LVR at the volatility the price actually realizes.

It is convenient to express \eqref{eq:lp_theta} through the depth density
\begin{equation}\label{eq:depth}
\delta(P)=\frac{dy}{d\ln P}=\frac{L\sqrt{P}}{2},
\end{equation}
the amount of numeraire the active liquidity absorbs per unit move in log price. Then
\begin{equation}\label{eq:lp_theta_depth}
\Theta_{\mathrm{LP}}(\sigma)=\tfrac12\sigma^2\,\delta.
\end{equation}

\section{A fee implied volatility}
\label{sec:fee_iv}

\subsection{Definition}

Fix a window $[t_0,t_1]$ of length $\Delta$ years. Let $V_{[t_0,t_1]}$ denote swap volume over the window, measured in numeraire, and let $\gamma$ denote the pool fee rate. Fee income over the window is $F_{[t_0,t_1]}=\gamma V_{[t_0,t_1]}$.

\begin{definition}[Fee implied volatility]\label{def:fee_iv}
The fee implied volatility $\sfee$ over $[t_0,t_1]$ is the constant volatility at which the aggregate theta of active liquidity equals realized fee income:
\begin{equation}\label{eq:matching}
\gamma V_{[t_0,t_1]}
=\int_{t_0}^{t_1}\Theta_{\mathrm{LP}}(\sfee)\,dt
=\tfrac12\sfee^2\int_{t_0}^{t_1}\delta_t\,dt.
\end{equation}
Equivalently, writing $\bar\delta$ for the time average of $\delta_t$ over the window,
\begin{equation}\label{eq:fee_iv}
\sfee^2=\frac{2\gamma\,V_{[t_0,t_1]}}{\Delta\,\bar\delta}.
\end{equation}
\end{definition}

Both sides of \eqref{eq:matching} are amounts of numeraire: realized fees on the left, theta accumulated over the window on the right. In \eqref{eq:fee_iv}, $V$ and $\bar\delta$ are both in numeraire, $\gamma$ is dimensionless and $\Delta$ is in years, so $\sfee$ is an annualized volatility with no further scaling. When the fee rate varies over the window, as with dynamic fees set by Uniswap v4 hooks, $\gamma V_{[t_0,t_1]}$ is replaced by realized fee income $F_{[t_0,t_1]}$ and \eqref{eq:fee_iv} reads $\sfee^2=2F_{[t_0,t_1]}/(\Delta\bar\delta)$.

\begin{remark}[Turnover form]\label{rem:turnover}
Since $2\delta=L\sqrt{P}$ is the virtual numeraire reserve of active liquidity, \eqref{eq:fee_iv} states that $\sfee^2$ is four times the fee rate times the annualized turnover of that reserve.
\end{remark}

\subsection{Relation to the tick liquidity form}
\label{subsec:tick_form}

Let $D$ denote the numeraire value of active liquidity over a band of log width $2\gamma$ around the current price. To first order in $\gamma$, $D=2\gamma\,\delta$, and \eqref{eq:fee_iv} becomes
\begin{equation}\label{eq:tick_form}
\sfee=2\gamma\sqrt{\frac{V_{[t_0,t_1]}}{\Delta\,\bar D}}.
\end{equation}
With daily windows, $\Delta=1/365$, this is the expression of Lambert \cite{Lambert2021OnChainVol} used by Panoptic, with tick liquidity read as numeraire depth over a band of log width $2\gamma$.

For the three original Uniswap v3 tiers, one tick spacing has log width $t_S\ln(1.0001)\approx 2\gamma$ by \eqref{eq:tS_fee_rel}, so depth over one tick spacing is the correct normalization for those tiers. The definition does not depend on that relation. For the 1 bp tier and for Uniswap v4 pools, one uses \eqref{eq:fee_iv} directly or measures $D$ over log width $2\gamma$.

\begin{remark}[Units]\label{rem:units}
Reading the tick liquidity in \eqref{eq:tick_form} as the raw Uniswap liquidity $L$ is not admissible. $L$ has units of the square root of the product of risky and numeraire amounts, so the ratio of volume to $L$ is not dimensionless and depends on token decimals. Expressed through $L$, \eqref{eq:fee_iv} reads $\sfee^2=4\gamma V_{[t_0,t_1]}/(\Delta\,\overline{L\sqrt{P}})$, in which the fee rate enters $\sfee$ with power one half, not power one.
\end{remark}

\subsection{Interpretation}

By construction, $\sfee$ is the volatility at which theta collected on active liquidity equals the fee income actually earned. Since theta and LVR are the same function of volatility (Proposition~\ref{prop:lp_theta}), $\sfee$ is also the volatility at which fee income would exactly offset LVR. Section~\ref{sec:identification} makes this breakeven reading precise and asks when $\sfee$ also estimates the volatility of the external reference price.

\section{When does fee implied volatility estimate external volatility?}
\label{sec:identification}

The fee implied volatility is defined from pool observables alone. This section asks what it measures and under which conditions it coincides with the volatility $\sigma$ of an external reference price, such as the price on a centralized exchange.

\subsection{Breakeven accounting}
\label{subsec:breakeven}

Milionis, Moallemi, Roughgarden and Zhang \cite{Milionis2022LVR} compare an LP with a rebalancing strategy that holds the same risky position but trades at the external price. Let $\Pi_{[t_0,t_1]}$ denote the LP's value relative to that strategy, including fee income, accumulated over the window. Up to corrections that vanish with the fee \cite{Milionis2023Fees},
\begin{equation}\label{eq:pnl_lvr}
\Pi_{[t_0,t_1]}=F_{[t_0,t_1]}-\int_{t_0}^{t_1}\tfrac12\sigma_t^2\,\delta_t\,dt,
\end{equation}
where $\sigma_t$ is the instantaneous volatility of the external price. By Proposition~\ref{prop:lp_theta}, the integrand is $\Theta_{\mathrm{LP}}(\sigma_t)$, the LVR of active liquidity.

\begin{proposition}[Breakeven volatility]\label{prop:breakeven}
Under \eqref{eq:pnl_lvr},
\begin{equation}\label{eq:pnl_breakeven}
\Pi_{[t_0,t_1]}=\tfrac12\int_{t_0}^{t_1}\delta_t\left(\sfee^2-\sigma_t^2\right)dt.
\end{equation}
The LP outperforms the rebalancing strategy over the window if and only if $\sfee^2$ exceeds the depth weighted average of instantaneous variance over the window.
\end{proposition}

\begin{proof}
Substitute $F_{[t_0,t_1]}=\tfrac12\sfee^2\int_{t_0}^{t_1}\delta_t\,dt$ from \eqref{eq:matching} into \eqref{eq:pnl_lvr}.
\end{proof}

Equation \eqref{eq:pnl_breakeven} is the AMM counterpart of the profit of a delta hedged short option in Black Scholes, $\tfrac12\int\Gamma S^2(\sigma_{\mathrm{imp}}^2-\sigma_t^2)\,dt$. Fees play the role of theta, collected at $\sfee$, and LVR plays the role of gamma cost, paid at realized volatility. We summarize the comparison by the fee to LVR ratio
\begin{equation}\label{eq:alpha}
\alpha_{[t_0,t_1]}
=\frac{F_{[t_0,t_1]}}{\int_{t_0}^{t_1}\tfrac12\sigma_t^2\delta_t\,dt}
=\frac{\sfee^2}{\bar\sigma^2},
\end{equation}
where $\bar\sigma^2$ is the depth weighted average variance over the window. The ratio is not bounded above: $\alpha>1$ exactly when LPs outperform rebalancing.

\subsection{A pure arbitrage benchmark}
\label{subsec:benchmark}

Let $\gbar=\gamma/(1-\gamma)$ denote the fee per unit of the amount that enters the curve, and $g=\ln(1+\gbar)$. A trader who buys the risky token from the pool pays a marginal price $P(1+\gbar)$, and a trader who sells it receives $P/(1+\gbar)$. With $S$ the external price and $Z=\ln(S/P)$ the log mispricing, arbitrage is therefore profitable exactly when $\lvert Z\rvert>g$. We consider the following benchmark.
\begin{enumerate}[leftmargin=3em, label=(A\arabic*)]
  \item The external log price has constant volatility $\sigma$ and constant drift $\mu$.
  \item Arbitrageurs trade continuously and without cost as soon as $\lvert Z\rvert$ exceeds $g$, moving the pool price just enough to return $Z$ to the boundary of the band $\lvert Z\rvert\le g$.
  \item All swaps are arbitrage swaps.
  \item Depth $\delta$ is constant over the window.
\end{enumerate}

\begin{proposition}[Identification under pure arbitrage]\label{prop:benchmark}
Under (A1) to (A4), with $Z$ in its stationary regime,
\begin{equation}\label{eq:benchmark}
\mathbb{E}\big[\sfee^2\big]=\kappa(\gamma)\,(1+\rho)\,\sigma^2,
\qquad
\kappa(\gamma)=\frac{\gbar}{\ln(1+\gbar)}=1+\tfrac12\gbar+O(\gbar^2),
\qquad
\rho=\frac{2\mu\,\mathbb{E}[Z]}{\sigma^2},
\end{equation}
with $\lvert\rho\rvert\le 2\lvert\mu\rvert g/\sigma^2$. The same limit holds almost surely for $\sfee^2$ as the window length grows.
\end{proposition}

The proof is in Appendix~\ref{app:benchmark}. It rests on two facts: under (A2) the total variation of the pool log price grows at rate $(\sigma^2+2\mu\mathbb{E}[Z])/(2g)$, and each unit of pool log price movement earns a fee $\gbar\delta$ (Proposition~\ref{prop:per_swap}). For the ETH/USDC 30 bps pool, $\kappa\approx 1.0015$, and even with a drift of $100\%$ per year and $\sigma^2=0.3$ per year, $\lvert\rho\rvert$ is about $0.02$.

\begin{corollary}\label{cor:benchmark}
In the benchmark of Proposition~\ref{prop:benchmark}:
\begin{enumerate}[leftmargin=2em, label=(\roman*)]
  \item $\sfee$ does not depend on the fee tier beyond the factor $\kappa(\gamma)$;
  \item arbitrage fee income equals $\kappa(1+\rho)$ times LVR, so by \eqref{eq:pnl_breakeven} LPs break even against rebalancing up to terms of order $\gamma$, consistent with the vanishing of arbitrage profits in the fast block limit \cite{Milionis2023Fees};
  \item $\sfee$ is not merely an activity index: it is a consistent estimator of $\sigma$.
\end{enumerate}
\end{corollary}

\begin{remark}[Linear fees, quadratic loss]\label{rem:linear_quadratic}
For a single price move of size $x$, arbitrage fee income is linear in $x$ while LVR is quadratic, which might suggest that the two cannot scale together. Over time, however, the number of arbitrage episodes per unit time is itself proportional to $\sigma^2/g$, and both legs end up proportional to $\sigma^2$.
\end{remark}

\subsection{Uninformed flow and frictions}
\label{subsec:noise}

Partition swaps into arbitrage swaps and all others, with volumes $V_A$ and $V_N$, and let $\sigma_{\mathrm{fee},A}^2=2\gamma V_A/(\Delta\bar\delta)$ be the fee implied variance of arbitrage flow alone. Since fees are linear in volume,
\begin{equation}\label{eq:noise_identity}
\sfee^2=\sigma_{\mathrm{fee},A}^2\left(1+\frac{V_N}{V_A}\right)
\end{equation}
holds exactly. Uninformed flow pays fees without adding to the variance of the external price, so it pushes $\sfee$ above $\sigma$. If its presence leaves arbitrage volume unchanged to first order, Proposition~\ref{prop:benchmark} gives
\begin{equation}\label{eq:alpha_noise}
\alpha\approx\kappa(\gamma)\left(1+\frac{V_N}{V_A}\right).
\end{equation}
The first order qualification matters: uninformed trades also move the pool price inside the band, which can either preempt or trigger arbitrage.

Frictions act in the opposite direction. Gas costs, priority fees and arbitrage latency mean that arbitrageurs wait until the mispricing exceeds some $g_{\mathrm{eff}}>g$. Modelling this as a wider band in (A2), the proof of Proposition~\ref{prop:benchmark} gives $\sigma_{\mathrm{fee},A}^2\approx\kappa\,(g/g_{\mathrm{eff}})\,\sigma^2$, so frictions bias $\sfee$ downward.

Discrete block times act in the same direction. When arbitrageurs can trade only at block times spaced $\tau$ apart, the band is monitored discretely, and the continuity correction of Broadie, Glasserman and Kou \cite{BroadieGlassermanKou1997} gives
\begin{equation}\label{eq:block_correction}
g_{\mathrm{eff}}\approx g+\beta\,\sigma\sqrt{\tau},
\qquad
\beta=\frac{\lvert\zeta(1/2)\rvert}{\sqrt{2\pi}}\approx 0.5826,
\end{equation}
with $\zeta$ the Riemann zeta function. The relevant parameter is $\sigma\sqrt{\tau}/g$, as in \cite{Milionis2023Fees}. In simulations of (A1) to (A3) with $\sigma=55\%$ per year and 12 second blocks, \eqref{eq:block_correction} matches the simulated ratio $\mathbb{E}[\sigma_{\mathrm{TV}}^2]/\sigma^2$ to within $0.2\%$ in all three original tiers. The implied downward bias in $\sfee$ is about $15\%$ for the 5 bps tier, $3\%$ for the 30 bps tier and under $1\%$ for the 1\% tier, so the continuous benchmark is accurate for the higher tiers and should not be applied to the 5 bps tier without this correction.

Combining these effects, with $g_{\mathrm{eff}}$ including both the block correction and the frictions above,
\begin{equation}\label{eq:alpha_full}
\alpha\approx\kappa(\gamma)\,\frac{g}{g_{\mathrm{eff}}}\left(1+\frac{V_N}{V_A}\right).
\end{equation}
Because the fee tier also shapes which traders use a pool \cite{HasbrouckRiveraSaleh2022}, differences in $\sfee$ across fee tiers of the same pair measure differences in $V_N/V_A$ and $g_{\mathrm{eff}}$ rather than differences in $\sigma$.

\subsection{What pool data alone can and cannot deliver}
\label{subsec:pool_only}

Equation \eqref{eq:alpha_full} locates the identification problem precisely. From pool observables alone, $\sfee$ is observed, but the factors $g/g_{\mathrm{eff}}$ and $V_N/V_A$ are not: classifying a swap as arbitrage requires knowing where the external price was when it executed. Without a reference price, $\sigma$ is therefore identified only up to these factors. With one, both become estimable. Arbitrage swaps are those that move the pool price toward the reference price from outside the band, and $g_{\mathrm{eff}}$ can be read from the distribution of the mispricing at which arbitrage swaps execute. The oracle free quantity $\sfee$ is thus exact as a breakeven volatility, and becomes an estimator of $\sigma$ once a reference price is used to correct for flow composition and frictions.

\section{Empirical illustration on the ETH/USDC 30 bps pool}
\label{sec:empirics}


We compute $\sfee$ from \eqref{eq:fee_iv} for the ETH/USDC 30 bps pool with daily windows, $\Delta=1/365$, and compare it with the realized volatility $\mathrm{RV}$ of ETH spot. Realized volatility is the natural first benchmark: by Proposition~\ref{prop:breakeven}, it is the volatility at which the LVR side of LP accounting is evaluated.

\begin{figure}[H]
\centering
\includegraphics[width=0.85\textwidth]{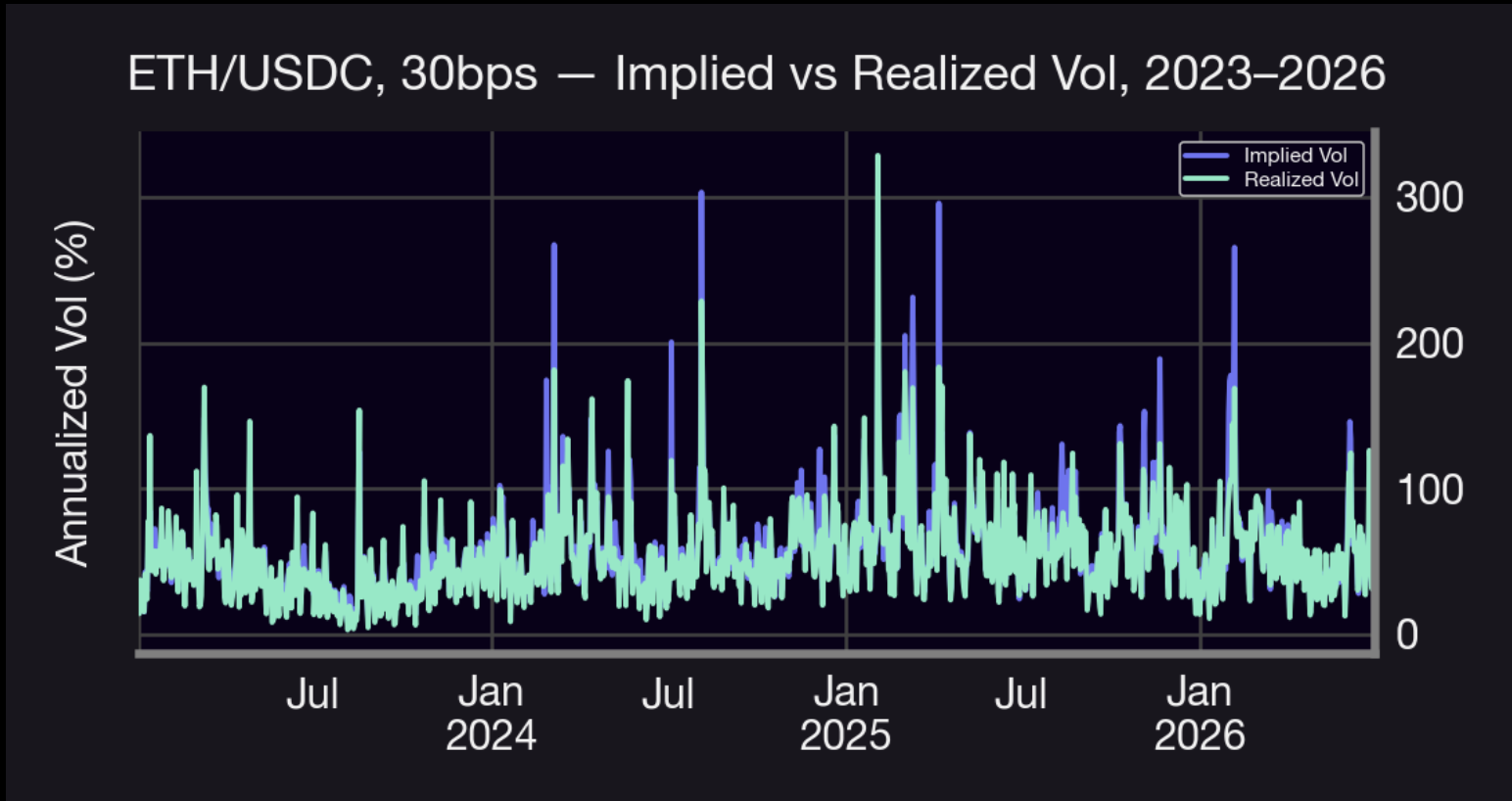}
\caption{$\sfee$ for the ETH/USDC 30 bps pool against realized volatility of ETH spot, daily observations from January 2023 to July 2026.}
\label{fig:sigma_fee_vs_rv}
\end{figure}

Figure~\ref{fig:sigma_fee_vs_rv} shows that $\sfee$ and $\mathrm{RV}$ move through the same compression and expansion regimes, with a correlation of $0.87$ in levels over the full window. Two cautions apply. Correlation in levels between two persistent volatility series is weak evidence on its own, and part of the comovement is mechanical, since same day arbitrage volume is generated by same day price moves. The informative tests are the correlation of changes and the ability of $\sfee$ to forecast realized volatility beyond its own lags.


Read through Proposition~\ref{prop:breakeven}, and to the extent that $\mathrm{RV}$ approximates realized volatility over the same window, the squared ratio
\[
\widehat{\alpha}_{\mathrm{RV}}(t)=\left(\frac{\sfee(t)}{\mathrm{RV}(t)}\right)^2
\]
is the ratio of fee income to LVR evaluated at realized variance, the empirical counterpart of \eqref{eq:alpha}. Over the sample it has mean $1.33$, median $1.18$ and interquartile range $[0.92,\,1.53]$, and it exceeds one on $68.26\%$ of days, which is the same event as $\sfee$ exceeding $\mathrm{RV}$. On roughly two days in three, fee income in this pool exceeded LVR at realized volatility, so LPs outperformed rebalancing.

Under \eqref{eq:alpha_full}, with $\kappa\approx 1.0015$ for this tier and a block correction \eqref{eq:block_correction} that lowers $\alpha$ by about $6\%$ at 12 second blocks, a median of $1.18$ is consistent with $V_N/V_A$ of roughly $0.25$, before any further reduction from gas costs and latency. These effects are not separated here; doing so requires the swap classification of Section~\ref{subsec:directions}. The excess of fees over LVR is the AMM counterpart of a variance risk premium: it is the compensation LPs require to supply liquidity, and in the benchmark it is paid by uninformed flow. Unlike its options market counterpart it is measured on realized rather than quoted quantities.

\begin{table}[H]
\centering
\begin{tabular}{lr}
\toprule
Metric & Value \\
\midrule
N (daily obs.) & 1276 \\
Mean $\sfee$ (\%) & 57.06 \\
Std $\sfee$ (\%) & 30.12 \\
Mean RV (\%) & 53.53 \\
Std RV (\%) & 28.83 \\
Corr($\sfee$, RV), levels & 0.87 \\
Mean spread, $\sfee$ minus RV (pp) & 3.53 \\
Std spread (pp) & 15.03 \\
Share of days $\sfee>\mathrm{RV}$ (\%) & 68.26 \\
Mean $\widehat{\alpha}_{\mathrm{RV}}$ & 1.33 \\
Median $\widehat{\alpha}_{\mathrm{RV}}$ & 1.18 \\
Std $\widehat{\alpha}_{\mathrm{RV}}$ & 0.74 \\
IQR $\widehat{\alpha}_{\mathrm{RV}}$ & $[0.92,\,1.53]$ \\
Min $\widehat{\alpha}_{\mathrm{RV}}$ & 0.26 \\
Max $\widehat{\alpha}_{\mathrm{RV}}$ & 11.33 \\
\bottomrule
\end{tabular}
\caption{Summary statistics for $\sfee$ against realized volatility, ETH/USDC 30 bps pool, 2023 to 2026. The share of days with $\widehat{\alpha}_{\mathrm{RV}}>1$ equals the share of days with $\sfee>\mathrm{RV}$.}
\label{tab:stats}
\end{table}


Figure~\ref{fig:alpha} plots $\widehat{\alpha}_{\mathrm{RV}}(t)$. Most of its mass sits just above one, punctuated by short lived spikes reaching $11.3$ that cluster around the volatility events of Figure~\ref{fig:sigma_fee_vs_rv}. Part of this shape is mechanical. Because $\widehat{\alpha}_{\mathrm{RV}}$ is a squared ratio, deviations below one are bounded while deviations above one are not, and any timing mismatch between the fee window and the $\mathrm{RV}$ window produces upward spikes precisely when volatility jumps. The median and interquartile range show that the premium is a typical day pattern rather than an artifact of the spikes: the 1st to 99th percentile trimmed mean of $1.30$ is close to the untrimmed mean of $1.33$.

\begin{figure}[H]
\centering
\includegraphics[width=1\textwidth]{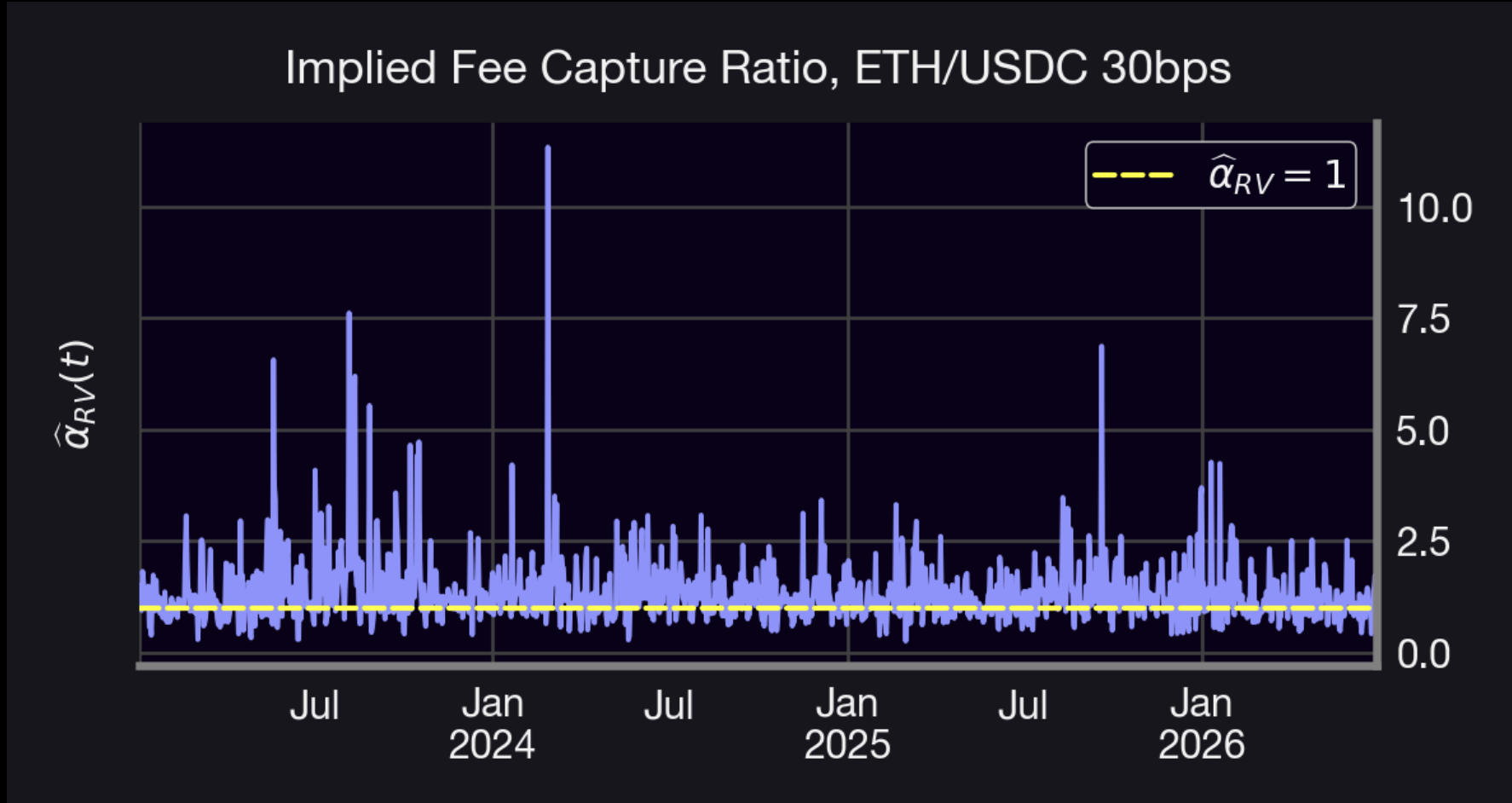}
\caption{Fee to LVR ratio at realized variance, $\widehat{\alpha}_{\mathrm{RV}}(t)=(\sfee(t)/\mathrm{RV}(t))^2$, over the same window as Figure~\ref{fig:sigma_fee_vs_rv}.}
\label{fig:alpha}
\end{figure}

Two caveats limit how much weight this illustration can bear. Realized volatility is a noisy estimate of the depth weighted variance in \eqref{eq:alpha}, so daily values of $\widehat{\alpha}_{\mathrm{RV}}$ are noisy estimates of $\alpha$. And this is a single pool over a single period. A systematic study across pools and fee tiers, with swaps classified against a reference price, is the natural next step.

\section{Implementation with on chain data}
\label{sec:implementation}

\subsection{Units and token ordering}
\label{subsec:units}

Pool prices are stored as token1 per token0 in raw units, and the numeraire may be either token: for the mainnet ETH/USDC pools, token0 is USDC and token1 is WETH. The depth density is one half of the virtual reserve of the numeraire token at the current price, that is $\delta=L\sqrt{P}/2$ when the numeraire is token1 and $\delta=L/(2\sqrt{P})$ when it is token0, with $P$ the raw pool price and the result converted to numeraire units with the token decimals. Volume is measured as swap input converted to numeraire at the execution price, so that fee income is exactly $\gamma V$.

\subsection{Windows and annualization}

Fix a window $[t_0,t_1]$ of $\Delta$ years, using a 365 day year since the market trades continuously. Then
\[
V_{[t_0,t_1]}=\sum_{\tau\in[t_0,t_1]}\mathrm{swapInput}_\tau,
\qquad
\bar\delta_{[t_0,t_1]}=\frac{1}{\Delta}\int_{t_0}^{t_1}\delta_u\,du,
\qquad
\sfee^2=\frac{2\gamma\,V_{[t_0,t_1]}}{\Delta\,\bar\delta_{[t_0,t_1]}},
\]
with $\delta_u$ computed from block level pool state. For daily windows, $\Delta=1/365$.

\subsection{A depth matched estimator}
\label{subsec:tv}

Definition~\ref{def:fee_iv} divides total fees by average depth. A more robust variant divides the fee paid by each swap by the depth that swap actually traded against. It rests on the following identity.

\begin{proposition}[Fee per unit of price movement]\label{prop:per_swap}
A swap that moves the pool log price monotonically from $\ln P$ to $\ln P'$ pays the fee
\begin{equation}\label{eq:per_swap}
F=\gbar\int_{\ln P}^{\ln P'}\delta\;\lvert d\ln P\rvert,
\end{equation}
with $\delta$ the depth along the path of the swap. When depth is constant along the swap, $F=\gbar\,\delta\,\lvert\Delta\ln P\rvert$.
\end{proposition}

\begin{proof}
By \eqref{eq:depth}, the numeraire value that enters the curve per unit move in log price is $\delta$. For a purchase of the risky token this is the numeraire paid in; for a sale, the risky amount paid in is worth $\delta$ per unit log move at the pool price. The input is this amount times $1+\gbar$, and the fee is $\gamma$ times the input, which equals $\gbar$ times the amount entering the curve.
\end{proof}

Dividing each fee by the depth it met and summing gives the depth matched fee implied variance
\begin{equation}\label{eq:tv_estimator}
\sigma_{\mathrm{TV}}^2=\frac{2\gbar}{\Delta}\,\mathrm{TV}_{[t_0,t_1]},
\qquad
\mathrm{TV}_{[t_0,t_1]}=\sum_{i}\lvert\Delta\ln P_i\rvert,
\end{equation}
where the sum runs over swaps in the window and $\Delta\ln P_i$ is the change in pool log price across swap $i$. Swaps move the price monotonically, including when they cross initialized ticks, and mints and burns do not move it, so $\mathrm{TV}$ is simply the total variation of the pool log price sampled at swaps.

The estimator has three useful features. It needs only the fee rate and the pool price after each swap, and no liquidity data at all. Liquidity that is minted and burned around a single swap, which a time average of $L$ misses, is accounted for automatically, because each swap's price move already reflects the depth it met. And Proposition~\ref{prop:benchmark} holds for $\sigma_{\mathrm{TV}}^2$ under (A1) to (A3) without (A4), since its proof works directly with total variation. The two estimators coincide when depth is constant and otherwise differ through the covariance between depth and price movement within the window. The noise identity \eqref{eq:noise_identity} carries over with $\mathrm{TV}$ split into its arbitrage and non arbitrage parts.

\subsection{Tick level attribution}

Uniswap v3 does not store cumulative volume per tick, and attributing swap flow to a narrow range requires reconstructing pool state over time \cite{Adams2021UniswapV3,UniswapDocsOracles}. Neither estimator above needs per range accounting: \eqref{eq:fee_iv} uses pool level depth, and \eqref{eq:tv_estimator} uses only the price path.

\subsection{Comparisons across pools}

By Section~\ref{subsec:noise}, differences in $\sfee$ across pools of the same pair measure differences in flow composition and arbitrage frictions rather than differences in $\sigma$. Cross pool comparisons are therefore informative about market structure, and comparisons across pairs call for the same reading.

\section{Scope and limits}
\label{sec:scope}

\subsection{Applications}

As an oracle free state variable, $\sfee$ supports the following uses.
\begin{enumerate}[leftmargin=2em]
  \item Scaling streamed premia: $\sfee$ is by construction the volatility at which premia matched to fee income compensate LVR.
  \item Monitoring LP profitability through the fee to LVR ratio $\alpha$, once an estimate of realized volatility is available.
  \item Comparing flow composition across pools and fee tiers.
  \item Informing collateral or margin buffers, with the caveat that $\sfee$ overstates $\sigma$ when uninformed flow is large and understates it when arbitrage is sluggish.
  \item Detecting bursts of activity relative to depth.
\end{enumerate}

\subsection{Where the CEX style interpretation breaks down}

$\sfee$ is computed from fee income realized over the window. It is a realized breakeven volatility, not a forward looking quote, and should not be read as the implied volatility quoted on centralized options venues. It coincides with the external volatility only in the benchmark of Section~\ref{subsec:benchmark}; outside it, its departure from $\sigma$ measures flow composition and frictions through \eqref{eq:alpha_full}. A comparison with an options implied volatility index compares a realized breakeven with a forward looking quote, and the gap between them mixes the variance risk premium with the noise and friction factors of \eqref{eq:alpha_full}.

\subsection{Research directions}
\label{subsec:directions}

\begin{enumerate}[leftmargin=2em]
  \item Classify swaps against a reference price to estimate $V_N/V_A$ and $g_{\mathrm{eff}}$, and form the corrected estimator
  \[
  \widehat{\sigma}=\sfee\left(\kappa(\gamma)\,\frac{g}{\widehat{g}_{\mathrm{eff}}}\left(1+\frac{\widehat{V}_N}{\widehat{V}_A}\right)\right)^{-1/2}.
  \]
  \item Test fee tier invariance on pairs traded in several tiers, starting with ETH/USDC at 5 and 30 bps, after applying the block correction \eqref{eq:block_correction}, which is large in the 5 bps tier.
  \item Compare $\sfee$ and $\sigma_{\mathrm{TV}}$ with a forward looking benchmark, and test their forecasting content for realized volatility.
\end{enumerate}

\section{Conclusion}
\label{sec:conclusion}

We have given the fee implied volatility of concentrated liquidity a derivation in which units, time and the fee tier are explicit, and an interpretation that is exact. Writing active liquidity as a short strip of options, the short maturity concentration of theta shows that its theta equals $\sigma^2L\sqrt{P}/4$ per unit time, which is its loss versus rebalancing. Matching this theta to fee income defines $\sfee^2=2\gamma V/(\Delta\bar\delta)$, the breakeven volatility of liquidity provision. With depth measured over a band of log width $2\gamma$ it is the tick liquidity formula in use since 2021, and in depth matched form it is a fee multiple of the realized total variation of the pool price.

Whether $\sfee$ measures the volatility of the external price is an identification question with a sharp answer. When all volume is arbitrage within the no trade band set by the fee, $\sfee$ recovers $\sigma$ up to a factor $1+O(\gamma)$ and the fee tier cancels. Uninformed flow pushes $\sfee$ up by the ratio of total to arbitrage volume, and frictions that widen the band push it down. From pool observables alone these factors are not separated; with a reference price they can be estimated. On the ETH/USDC 30 bps pool the fee to LVR ratio has a median of $1.18$, consistent with uninformed flow paying LPs a moderate premium over LVR.

For Panoptic, $\sfee$ is a natural oracle free state variable. It is the volatility at which streamed premia exactly compensate LVR, and its departure from external volatility is a direct measure of how much LP income comes from uninformed flow.

\newpage
\appendix

\section{Theta mass and localization}
\label{app:theta}

This appendix computes the integral of theta over spot price and proves Proposition~\ref{prop:dirac_theta}. Proposition~\ref{prop:dirac_strike} follows from the same computations with $u=\ln(K/S)$.

\begin{proposition}
For theta given by
\[
\theta(S,t)
=
\frac{S\sigma}{\sqrt{8\pi t}}
\exp\left(
-
\frac{\left(\ln(S/K)+\frac{\sigma^2 t}{2}\right)^2}{2\sigma^2 t}
\right),
\]
one has
\[
\int_0^\infty \theta(S,t)\,dS
=
\frac{K^2\sigma^2}{2}\exp\!\left(\sigma^2 t\right)
\longrightarrow
\frac{K^2\sigma^2}{2}
\qquad
\text{as } t\to 0.
\]
\end{proposition}

\begin{proof}
Set $u=\ln(S/K)$, so that $S=Ke^u$ and $dS=Ke^u\,du$. Write $a=\sigma^2 t$. Then
\begin{align*}
\int_0^\infty \theta(S,t)\,dS
&=
\int_{-\infty}^{\infty}
\frac{Ke^u\sigma}{\sqrt{8\pi t}}
\exp\left(
-
\frac{\left(u+\frac{a}{2}\right)^2}{2a}
\right)
Ke^u\,du \\
&=
\frac{K^2\sigma}{\sqrt{8\pi t}}
\int_{-\infty}^{\infty}
\exp\left(
-
\frac{\left(u+\frac{a}{2}\right)^2}{2a}
+
2u
\right)\,du.
\end{align*}
Let $v=u+a/2$, so that $2u=2v-a$. The exponent becomes
\[
-\frac{v^2}{2a}+2v-a
=
-\frac{(v-2a)^2}{2a}+a,
\]
so that
\[
\int_{-\infty}^{\infty}
\exp\left(
-
\frac{\left(u+\frac{a}{2}\right)^2}{2a}
+
2u
\right)\,du
=
\exp(a)\int_{-\infty}^{\infty}
\exp\left(-\frac{(v-2a)^2}{2a}\right)\,dv
=
\exp(a)\sqrt{2\pi a}.
\]
Substituting back,
\[
\int_0^\infty \theta(S,t)\,dS
=
\frac{K^2\sigma}{\sqrt{8\pi t}}\cdot\exp(a)\sqrt{2\pi a}
=
K^2\sigma\cdot\sqrt{\frac{2\pi a}{8\pi t}}\cdot\exp(a)
=
K^2\sigma\cdot\frac{\sigma}{2}\cdot\exp(a),
\]
using $a=\sigma^2 t$ in the last step. Hence
\[
\int_0^\infty \theta(S,t)\,dS
=
\frac{K^2\sigma^2}{2}\exp(\sigma^2 t),
\]
which converges to $\frac{K^2\sigma^2}{2}$ as $t\to 0$.
\end{proof}

\paragraph{Proof of Proposition~\ref{prop:dirac_theta}.}
Let $f$ be a smooth compactly supported test function. Using the same substitution $u=\ln(S/K)$, one obtains
\begin{align*}
\int_0^\infty f(S)\,\theta(S,t)\,dS
&=
\frac{K^2\sigma}{\sqrt{8\pi t}}
\int_{-\infty}^{\infty}
f(Ke^u)
\exp\left(
-
\frac{\left(u+\frac{\sigma^2 t}{2}\right)^2}{2\sigma^2 t}
+
2u
\right)\,du.
\end{align*}
Define $\varepsilon^2=\sigma^2 t$. The Gaussian factor localizes the integral near $u=0$ as $\varepsilon\to 0$. Since $f(Ke^u)e^{2u}\to f(K)$ uniformly on shrinking neighborhoods of zero, and the Gaussian kernel is integrable and normalized up to a constant factor, dominated convergence gives
\[
\int_0^\infty f(S)\,\theta(S,t)\,dS
\longrightarrow
\frac{K^2\sigma^2}{2}\,f(K),
\]
that is, $\theta(S,t)\rightharpoonup \frac{K^2\sigma^2}{2}\,\delta(S-K)$. \qed

\section{Proof of Proposition~\ref{prop:benchmark}}
\label{app:benchmark}

Let $X_t=\ln S_t$ and $Y_t=\ln P_t$. Under (A2) the pool log price moves only when the mispricing $Z=X-Y$ sits on the boundary of the band, so
\[
Y_t=Y_0+\ell^{+}_t-\ell^{-}_t,
\]
where $\ell^{+}$ and $\ell^{-}$ are continuous and nondecreasing, and increase only on $\{Z=g\}$ and $\{Z=-g\}$ respectively. Thus $Z$ is a Brownian motion with drift $\mu$ and volatility $\sigma$ reflected at $\pm g$, and the total variation of the pool log price over $[0,t]$ is $\mathrm{TV}_t=\ell^{+}_t+\ell^{-}_t$. Since $dZ_t=\mu\,dt+\sigma\,dW_t-d\ell^{+}_t+d\ell^{-}_t$, Itô's formula applied to $Z^2$ gives
\[
d\big(Z_t^2\big)=2Z_t(\mu\,dt+\sigma\,dW_t)+\sigma^2\,dt-2g\,d\ell^{+}_t-2g\,d\ell^{-}_t,
\]
using $Z=g$ on the support of $d\ell^{+}$ and $Z=-g$ on the support of $d\ell^{-}$. In the stationary regime $\mathbb{E}[Z_t^2]$ is constant, and the stochastic integral has mean zero because $Z$ is bounded, so
\[
\mathbb{E}[\mathrm{TV}_t]=t\,\frac{\sigma^2+2\mu\,\mathbb{E}[Z]}{2g}.
\]
By Proposition~\ref{prop:per_swap} with constant depth (A4), fee income over a window of length $\Delta$ is $F=\gbar\,\delta\,\mathrm{TV}$, so $\mathbb{E}[F]=\gbar\,\delta\,\Delta\,(\sigma^2+2\mu\mathbb{E}[Z])/(2g)$, and \eqref{eq:fee_iv} gives
\[
\mathbb{E}\big[\sfee^2\big]=\frac{2\,\mathbb{E}[F]}{\Delta\,\delta}=\frac{\gbar}{g}\,(1+\rho)\,\sigma^2=\kappa(\gamma)\,(1+\rho)\,\sigma^2.
\]
The bound on $\rho$ follows from $\lvert Z\rvert\le g$. For the almost sure statement, integrate the Itô identity over $[0,t]$ and divide by $t$: the term $Z_t^2/t$ vanishes because $Z$ is bounded, the stochastic integral divided by $t$ vanishes almost surely, and $\frac1t\int_0^t Z_s\,ds\to\mathbb{E}[Z]$ by the ergodic theorem for the reflected process. For $\sigma_{\mathrm{TV}}^2$ in \eqref{eq:tv_estimator}, the same argument applies without (A4), since $\sigma_{\mathrm{TV}}^2=2\gbar\,\mathrm{TV}/\Delta$ directly. \qed

\section*{Acknowledgments}

I thank the Panoptic team for numerous discussions on premium mechanics, streaming theta, and range level fee accounting that shaped the framing of this paper. I am especially grateful to Guillaume Lambert, inventor of the Panoptic protocol, for his guidance on the economics of narrow range liquidity positions and for feedback on earlier drafts of this work.

\bibliographystyle{plain}
\bibliography{refs}

@techreport{Adams2021UniswapV3,
  author       = {Hayden Adams and Noah Zinsmeister and Moody Salem and River Keefer and Dan Robinson},
  title        = {Uniswap v3 Core},
  institution  = {Uniswap Labs},
  year         = {2021},
  note         = {Technical report},
  url          = {https://uniswap.org/whitepaper-v3.pdf}
}

@article{Lambert2022Panoptic,
  author       = {Guillaume Lambert and Jesper Kristensen},
  title        = {Panoptic: the perpetual, oracle free options protocol},
  journal      = {arXiv preprint},
  eprint       = {2204.14232},
  archivePrefix= {arXiv},
  primaryClass = {q-fin.PR},
  year         = {2022},
  note         = {Version 1.3.1, June 2023}
}

@inproceedings{Hashemseresht2022Concentrated,
  author       = {Saleh Hashemseresht and Mohsen Pourpouneh},
  title        = {Concentrated Liquidity Analysis in Uniswap v3},
  booktitle    = {DeFi'22: Proceedings of the 2022 ACM CCS Workshop on Decentralized Finance and Security},
  pages        = {63--70},
  year         = {2022},
  publisher    = {ACM},
  doi          = {10.1145/3560832.3563438}
}

@inproceedings{Fan2023StrategicLP,
  author       = {Zhou Fan and Francisco Marmolejo-Coss{\'i}o and Daniel J. Moroz and Michael Neuder and Rithvik Rao and David C. Parkes},
  title        = {Strategic Liquidity Provision in Uniswap v3},
  booktitle    = {5th Conference on Advances in Financial Technologies (AFT 2023)},
  series       = {LIPIcs},
  volume       = {282},
  pages        = {25:1--25:22},
  year         = {2023},
  publisher    = {Schloss Dagstuhl},
  doi          = {10.4230/LIPIcs.AFT.2023.25},
  eprint       = {2106.12033},
  archivePrefix= {arXiv}
}

@misc{UniswapDocsConcentrated,
  author       = {{Uniswap Labs}},
  title        = {Concentrated Liquidity},
  howpublished = {\url{https://docs.uniswap.org/concepts/protocol/concentrated-liquidity}},
  note         = {Accessed 2025}
}

@misc{UniswapDocsOracles,
  author       = {{Uniswap Labs}},
  title        = {TWAP Oracles},
  howpublished = {\url{https://docs.uniswap.org/concepts/protocol/oracle}},
  note         = {Accessed 2025}
}

@misc{PanopticDocsPerps,
  author       = {{Panoptic Labs}},
  title        = {Perpetual Options on Panoptic},
  howpublished = {\url{https://panoptic.xyz/docs/trading/perpetual-options}},
  note         = {Accessed 2025}
}

@article{BlackScholes1973,
  author       = {Fischer Black and Myron Scholes},
  title        = {The Pricing of Options and Corporate Liabilities},
  journal      = {Journal of Political Economy},
  volume       = {81},
  number       = {3},
  pages        = {637--654},
  year         = {1973}
}

@article{Merton1973,
  author       = {Robert C. Merton},
  title        = {Theory of Rational Option Pricing},
  journal      = {The Bell Journal of Economics and Management Science},
  volume       = {4},
  number       = {1},
  pages        = {141--183},
  year         = {1973}
}

@book{Lighthill1960,
  author       = {M. J. Lighthill},
  title        = {An Introduction to Fourier Analysis and Generalised Functions},
  publisher    = {Cambridge University Press},
  year         = {1960}
}

@article{Bibinger2014LOBVol,
  author       = {Markus Bibinger and Moritz Jirak and Markus Reiss},
  title        = {Improved volatility estimation based on limit order books},
  journal      = {SFB 649 Discussion Paper},
  number       = {2014-053},
  year         = {2014},
  institution  = {Humboldt University of Berlin}
}

@article{Naes2006VolumeVol,
  author       = {Randi N{\ae}s and Johannes A. Skjeltorp},
  title        = {Order book characteristics and the volume volatility relation: Empirical evidence from a limit order market},
  journal      = {Journal of Financial Markets},
  volume       = {9},
  number       = {4},
  pages        = {408--432},
  year         = {2006}
}

@article{Cartea2023PredictableLoss,
  author       = {Alvaro Cartea and Fay{\c c}al Drissi and Marcello Monga},
  title        = {Decentralised finance and automated market making: predictable loss and optimal liquidity provision},
  journal      = {arXiv preprint},
  eprint       = {2309.08431},
  archivePrefix= {arXiv},
  primaryClass = {q-fin.TR},
  year         = {2023}
}

@misc{Castle2023DeFiOptions,
  author       = {Castle Capital Research},
  title        = {Demystifying Options: A 0 to 1 Guide about DeFi Options},
  howpublished = {\url{https://chronicle.castlecapital.vc/p/demystifying-options-0-1-guide-defi-options}},
  year         = {2023},
  note         = {Online article}
}

@misc{Polygon2022OnChainOptions,
  author       = {Polygon Labs},
  title        = {A Game of Premiums Through a Myriad of Complexities: A Deep Dive into On Chain Option Protocols},
  howpublished = {\url{https://polygontech.medium.com/a-game-of-premiums-through-a-myriad-of-complexities-a-deep-dive-into-on-chain-option-protocols-d9619fe99278}},
  year         = {2022},
  note         = {Online article}
}

@book{Gatheral2006VolSurface,
  author    = {Jim Gatheral},
  title     = {The Volatility Surface: A Practitioner's Guide},
  publisher = {Wiley},
  year      = {2006},
  address   = {Hoboken, NJ}
}

@misc{Lambert2021OnChainVol,
  author       = {Guillaume Lambert},
  title        = {On chain volatility and {Uniswap} v3},
  howpublished = {Medium},
  month        = nov,
  year         = {2021}
}

@misc{Lambert2021PerpetualOptions,
  author       = {Guillaume Lambert},
  title        = {{Uniswap} v3 {LP} tokens as perpetual put and call options},
  howpublished = {Medium},
  month        = jun,
  year         = {2021}
}

@misc{KhaldounPanopticIV,
  author       = {Amy Khaldoun},
  title        = {From {CEX} to {DEX}: A new formulation for implied volatility, part 2},
  howpublished = {Panoptic Research}
}

@article{Milionis2022LVR,
  author  = {Jason Milionis and Ciamac C. Moallemi and Tim Roughgarden and Anthony Lee Zhang},
  title   = {Automated market making and loss versus rebalancing},
  journal = {arXiv preprint arXiv:2208.06046},
  year    = {2022}
}

@article{Milionis2023Fees,
  author  = {Jason Milionis and Ciamac C. Moallemi and Tim Roughgarden},
  title   = {Automated market making and arbitrage profits in the presence of fees},
  journal = {arXiv preprint},
  year    = {2023}
}

@article{CarrMadan2001,
  author  = {Peter Carr and Dilip Madan},
  title   = {Optimal positioning in derivative securities},
  journal = {Quantitative Finance},
  volume  = {1},
  number  = {1},
  year    = {2001}
}

@misc{Clark2020Replicating,
  author       = {Joseph Clark},
  title        = {The replicating portfolio of a constant product market},
  howpublished = {SSRN working paper},
  year         = {2020}
}

@misc{HasbrouckRiveraSaleh2022,
  author       = {Joel Hasbrouck and Thomas J. Rivera and Fahad Saleh},
  title        = {The need for fees at a {DEX}: How increases in fees can increase {DEX} trading volume},
  howpublished = {Working paper},
  year         = {2022}
}

\end{document}